\pdfoutput=1
\RequirePackage{ifpdf}
\ifpdf 
\documentclass[pdftex]{sigma}
\else
\documentclass{sigma}
\fi

\def\al{\alpha}
\def\be{\beta}
 \def\Ga{\Gamma}

\def\te{\theta} 
\def\om{\omega}

\def\wh{\widehat}

\def\fracpd#1#2{\frac{\partial #1}{\partial #2}}

\begin{document}

\allowdisplaybreaks

\renewcommand{\thefootnote}{$\star$}

\newcommand{\arXivNumber}{1509.07493}

\renewcommand{\PaperNumber}{010}

\FirstPageHeading

\ShortArticleName{Quasi-Bi-Hamiltonian Structures of the 2-Dimensional Kepler Problem}

\ArticleName{Quasi-Bi-Hamiltonian Structures\\ of the 2-Dimensional Kepler Problem\footnote{This paper is a~contribution to the Special Issue
on Analytical Mechanics and Dif\/ferential Geometry in honour of Sergio Benenti.
The full collection is available at \href{http://www.emis.de/journals/SIGMA/Benenti.html}{http://www.emis.de/journals/SIGMA/Benenti.html}}}

 \Author{Jose F.~{CARI\~NENA} and Manuel F.~{RA\~NADA}}
 \AuthorNameForHeading{J.F.~Cari\~nena and M.F.~Ra\~nada}
 \Address{Departamento de F\'{\i}sica Te\'orica and IUMA, Universidad de Zaragoza, 50009 Zaragoza, Spain}
 \Email{\href{mailto:jfc@unizar.es}{jfc@unizar.es}, \href{mailto:mfran@unizar.es}{mfran@unizar.es}}

\ArticleDates{Received September 29, 2015, in f\/inal form January 25, 2016; Published online January 27, 2016}

\Abstract{The existence of quasi-bi-Hamiltonian structures for the Kepler problem is studied.
We f\/irst relate the superintegrability of the system with the existence of two complex functions
endowed with very interesting Poisson bracket properties and then we prove the existence
of a quasi-bi-Hamiltonian structure by making use of these two functions.
The paper can be considered as divided in two parts.
In the f\/irst part a quasi-bi-Hamiltonian structure is obtained by making use of polar coordinates
and in the second part a new quasi-bi-Hamiltonian structure is obtained by making use of the
separability of the system in parabolic coordinates.}


\Classification{37J15; 37J35; 70H06; 70H33}

\renewcommand{\thefootnote}{\arabic{footnote}}
\setcounter{footnote}{0}

\section{Introduction}

\subsection{A geometric introduction}

In dif\/ferential geometric terms, the phase space $M$ of a Hamiltonian system
is the $2n$-dimensional cotangent bundle $M=T^*Q$ of the $n$-dimensional
conf\/iguration space~$Q$.
Cotangent bundles are manifolds endowed, in a natural or canonical
way, with a symplectic structure~$\om_0$; if $\{(q_i) \,|\, i=2,\ldots,n\}$ are local
coordinates in $Q$ and
$\{(q_j,p_j);\,j=1,2,\dots,n\}$ the induced coordinates in $T^*Q$, then $\om_0$ is given by
\begin{gather*}
 \om_0 = dq_j \wedge dp_j , \qquad \om_0 = - d {\te_0}
 , \qquad \te_0 = p_j dq_j ,
\end{gather*}
(we write all the indices as subscripts and summation convention on the repeated index is used).
Given a dif\/ferentiable function $F$ in $T^*Q$, $F=F(q,p)$, the vector f\/ield~$X_F$ uniquely
def\/ined as the solution of the equation
\begin{gather*}
 i(X_F) \om_0 = dF
\end{gather*}
is called the Hamiltonian vector f\/ield of the function $F$.
In particular, given a Hamiltonian $H=H(q,p)$, the dynamics is given by the corresponding
Hamiltonian vector f\/ield~$X_H$, that is, $i(X_H) \om_0 = dH$.

 A vector f\/ield $\Gamma\in\mathfrak{X}(T^*Q)$ is Hamiltonian if there is a function $H$ such that $\Ga=X_H$, i.e., $i(\Ga)\om_0=dH$, and locally-Hamiltonian when $i(\Ga)\omega_0$ is a closed 1-form.
This is equivalent to $\Ga$-invariance of~$ \omega_0$, i.e., $\mathcal{L}_\Ga\omega_0=0$.

A system of dif\/ferential equations is called bi-Hamiltonian if it can be written in
two dif\/ferent ways in Hamiltonian form.
Suppose a manifold $M$ equipped with two dif\/ferent symplectic structures $\om_0$ and $\om_1$.
A vector f\/ield $\Ga$ on $T^*Q$ is said to be a bi-Hamiltonian vector f\/ield if
it is Hamiltonian with respect to both symplectic structures, i.e.,
\begin{gather*}
 i(\Ga) \om_0 = d H_0 \qquad{\rm and}\qquad i(\Ga) \om_1 = d H_1 .
\end{gather*}
The two functions, $H_0$ and $H_1$, are integrals of motion for $\Gamma$.
A weaker form of bi-Hamiltonian system is when the only symplectic form is the f\/irst one
($\om_1$ is a closed but nonsymplectic 2-form).

We point out that an important example of bi-Hamiltonian system is the rational harmonic oscillator (non-central harmonic oscillator with rational ratio of frequencies)~\cite{CarMarRan02}
\begin{gather*}
 H = \frac12 \bigl( p_x^2 + p_y^2 \bigr) + \frac12 \alpha_0^2\bigl( m^2 x^2 + n^2 y^2 \bigr) .
\end{gather*}

A symplectic form determines a Poisson bivector $\Lambda$ that satisf\/ies the vanishing of the Schouten bracket $[\Lambda,\Lambda]=0$ (this property is equivalent to the Jacobi identity); we note that there exist non-constant rank Poisson structure not related with symplectic structures.
 The compatibility condition between two dif\/ferent Poisson structures,
 $\Lambda_0$ and $\Lambda_1$, means that the linear combination $\Lambda_{\lambda} = \Lambda_0 - {\lambda} \Lambda_1$ is a Poisson pencil, that is, it is a~Poisson bivector for every value of ${\lambda}$;
therefore the corresponding bracket $\{\cdot,\cdot\}_{\lambda}=\{\cdot,\cdot\}_0 - {\lambda} \{\cdot,\cdot\}_1$ is a pencil of Poisson brackets.

Bi-Hamiltonian systems are systems endowed with very interesting properties but, in general, it is quite dif\/f\/icult to f\/ind a bi-Hamiltonian formulation for a~given Hamiltonian vector f\/ield, and for this reason it is useful to introduce the concept of quasi-bi-Hamiltonian system.

\begin{definition}\label{definition1}
A vector f\/ield $X$ on a symplectic manifold $(M,\omega)$ is called quasi-Hamiltonian if there exists a (nowhere-vanishing) function~$\mu$ such that~$\mu X$ is a Hamiltonian vector f\/ield
\begin{gather*}
 \mu X \in {\mathfrak X}_H(M).
\end{gather*}
Thus $i(\mu X)\omega = dh$ for some function $h$.
\end{definition}

We call $\mu$ an integrating factor of the quasi-Hamiltonian system, because it is an integrating factor for the 1-form $i(X)\omega$, and we note that in this case the function~$h$ is a f\/irst integral of~$X$.
Note that this condition can alternatively be written as as $i(X)(\mu \omega)=dh$, but the point is that the 2-form $\mu\omega$ is not closed in the general case.

The scarcity of bi-Hamiltonian systems leads to relax the condition for the vector f\/ield being Hamiltonian to a simpler situation of quasi-Hamiltonian with respect to the second symplectic structure.

\begin{definition}\label{definition2}
A Hamiltonian vector f\/ield $X$ on a symplectic manifold $(M,\omega)$ is called quasi-bi-Hamiltonian if there exist another symplectic structure~$\omega_1$, and a nowhere-vanishing func\-tion~$\mu$, such that $\mu X$ is a Hamiltonian vector f\/ield with respect to~$\omega_1$.
\end{definition}

This concept was f\/irst introduced in \cite{BrouCaboz96} in the particular case of systems with two degrees of freedom
and it was quickly extended in \cite{MorTondo97,MorTondo98} for a~higher-dimensional systems.
Some recent papers considering properties of this particular class of systems are
\cite{Blasz09, BoualBrouzRak06, BoualBrouz08, BrouCaboz96, CarGuhRan07, CarGuhRan08, CarMarRan02, CrampSar01, CrampSar02, MorTondo97, MorTondo98, Rakot11, ZengMa99}.

The nondegeneracy of the canonical form $\om_0$ provides a vector bundle isomorphism $\wh{\om_0}$ of $T(T^*Q)$ on $T^*(T^*Q)$, inducing an identif\/ication of vector f\/ields and 1-forms on the phase space.
A consequence is that the pair $(\omega_0,\omega_1)$ determines a $(1,1)$ tensor f\/ield $R$ def\/ined as
$R = \wh{\om_0}^{-1} \circ \wh{\omega_{1}}$, that is,
\begin{gather*}
 \omega_1(X,Y) = \omega_0(RX,Y) , \qquad \forall \, X,Y \in \mathfrak{X}(T^*Q) .
\end{gather*}
Note that in the def\/inition of $R$ the 2-form $\om_1$ is necessarily closed but it can be nonsymplectic.
If $\Ga$ is bi-Hamiltonian with respect to $(\omega_0,\omega_1)$ then $R$ is $\Gamma$-invariant, that is, ${\mathcal L}_{\Gamma}R = 0$, where~${\mathcal L}$ denotes the Lie derivative (this means that the characteristic polynomial of~$R$ is an inva\-riant for~$\Ga$, and consequently the coef\/f\/icients of the polynomial are constants of motion).
The Nijenhuis tensor $N_R$ of the tensor f\/ield $R$ is def\/ined by
\begin{gather*}
 N_R(X,Y) = R^2[X,Y] + [RX,RY] - R[RX,Y] - R[X,RY] .
\end{gather*}
It has been proved that if $\Ga$ is a Hamiltonian dynamical system of the type described above and such that
(i) The tensor $N_R$ of $R$ vanishes,
(ii) The tensor f\/ield $R$ has~$n$ distinct eigenfunctions (that is, they are maximally distinct),
then the eigenfunctions of~$R$ are in involution and the system is therefore completely integrable \cite{CrMarRub83, CrampSarTh00, DeFVilMar84, Fernandes94, Marmo83}.
It is important to note that the eigenvalues of $R$ are constants of motion for $\Ga$ even in the case that the two properties~(i) and~(ii) are not satisf\/ied (but then the eigenfunctions are not in involution).

It is known that the Liouville formalism characterize the Hamiltonians that are integrable but it does not provide a method for obtaining the integrals of motion; therefore it has been necessary to elaborate dif\/ferent methods for obtaining constants of motion (Hamilton--Jacobi separability, Lax pairs formalism, Noether symmetries, Hidden symmetries, etc); the existence of a bi-Hamiltonian structure with the above two mentioned properties (Nijenhuis torsion condition and maximally distinct eigenvalues) can be considered as method to establish the Liouville integrability of a system; because of this, these two properties are frequently included in the def\/inition of a bi-Hamiltonian structure.

Most of systems admitting a bi-Hamiltonian structure are separable systems; so these two properties (separability and double Hamiltonian structure) are properties very close related (see~\cite{Fernandes94} for a detailed discussion of this question and~\cite{TempTond12} for the case of multiple separability).
Quasi-bi-Hamiltonian systems are very less known than the bi-Hamiltonian ones but it seems that they are also related with separability.
Nevertheless, in this case the tensor f\/ield~$R$ is not $\Gamma$-invariant and the eigenvalues of the tensor f\/ield $R$ are not constants of motion.
An interesting property is that it was shown in~\cite{BrouCaboz96} that in the particular case of two degrees of freedom the function~$\mu^2/\det R$ is a constant of the motion.

The potential of the Kepler problem is spherically symmetric and therefore it admits alternative Lagrangians (the existence of alternative Lagrangians for central potentials is studied in \cite{CrampPrin88, HenShep82, RanJmp91}); recall that if there exist alternative Lagrangian descriptions then one can f\/ind non-Noether constants of motion \cite{CarIb83}.
This system has been studied as a bi-Hamilltonian system by making use of dif\/ferent approaches;
 Rauch-Wojciechowski proved the existence of a~bi-Hamiltonian formulation but introducing an extra variable so that the phase space is odd-dimensional and the Poisson brackets are degenerate~\cite{RWojcie91} and more recently~\cite{GrigTsi15} a~bi-Hamiltonian formulation for the perturbed Kepler problem has also been studied by making use of Delaunay-type variables.

\subsection{Structure and purpose of the paper}

Now in this paper we will analyze certain properties of the Kepler problem related with the existence of quasi-bi-Hamiltonian structures.
The following two points summarize the contents of the paper.
\begin{itemize}\itemsep=0pt
\item
First, we will study the existence of certain complex functions with interesting Poisson properties and
then we will prove that the superintegrability of the system is directly related with the properties of these complex functions.
As is well-known this system is multi-separable (it separates in both polar and parabolic coordinates); so
we f\/irst present the study making use of polar coordinates and then we undertake a similar study by employing parabolic coordinates (the parabolic complex functions are dif\/ferent from the polar ones).

\item
Second, we prove that the above mentioned complex functions determine the existence of several quasi-bi-Hamiltonian structures. This is done in two steps: f\/irst with complex 2-forms and then with several real 2-forms. The properties of these geometric structures and of the associated recursion operators are analyzed.

 It is important to note that this study is concerned with the existence of quasi-bi-Hamil\-tonian structures (instead of bi-Hamiltonian). So we recall that if
$ i(\Gamma)\omega_1={\lambda}dH_1$ then ${\mathcal X}_{\Gamma}\omega_1
 = d{\lambda}\wedge dH_1\ne 0 $.
Consequently the tensor f\/ield $R$ is not $\Gamma$-invariant and the eigenvalues of $R$ are not constants of motion.
\end{itemize}

We must clearly say that the structures obtained by this method (wedge product of the dif\/ferentials of complex functions) do not satisfy the above mentioned Nijenhuis torsion condition.
So perhaps it is convenient to name them as weak quasi-bi-Hamiltonian structures (in opposition to strong structures satisfying the Nijenhuis condition).
Nevertheless, the purpose in this paper is not to prove the integrability of a system as consequence of a bi-Hamiltonian structure, since it is perfectly known that the Kepler problem is not only integrable but also superintegrable. The purpose is to study new and interesting properties of the Kepler problem.
 In fact, it has been proved that if a dynamical vector f\/ield satisf\/ies certain properties (existence of canonoid transformations~\cite{CarRan88, CarRan90} or existence of non-symplectic symmetries~\cite{Ran00,Ran05}) then it is Hamiltonian with respect to two dif\/ferent structures without satisfying necessarily the Nijenhuis condition.

 The structure of the paper is as follows:
 In Section~\ref{section2} we study the Kepler problem by making use of polar coordinates $(r,\phi)$; we relate the superintegrability of the system with the existence of two complex functions~$M_r $ and~$N_{\phi}$ endowed with very interesting Poisson bracket properties and then we prove the existence of a~quasi-bi-Hamiltonian structure making use of these two functions.
 Then in Section~\ref{section3} we consider once more the same system but in terms of parabolic coordinates~$(a,b)$ and we obtain a new quasi-bi-Hamiltonian structure (dif\/ferent to the previous one) making use of a similar technique but with new complex functions~$M_a$ and~$M_b$. Finally in Section~\ref{section4} we make some f\/inal comments.

\section{From superintegrability to quasi-bi-Hamiltonian structures}\label{section2}

After these rather general comments we restrict our study to the Kepler problem in the Euclidean plane that, as it is well known, is superintegrable and multiseparable (polar and parabolic coordinates).

 Let us f\/irst notice that in some cases the two-dimensional Euclidean systems possess certain interesting properties.
For example, if the potential $V(x,y)$ takes the form $V=A(u)+B(v)$, $u=x+y$, $v=x-y$, then it admits a new Hamiltonian structure (and also a Lax pair) \cite{McLSmir00}; unfortunately the new structure is in most of cases constant (we mean that the new Poisson bracket is determined by a symplectic form as $\om_1 = dx \wedge dp_y + dy \wedge dp_x$).

 In the next paragraphs we will relate the existence of bi-Hamiltonian structures with the properties of the two complex functions with interesting Poisson bracket properties.

\subsection{Complex functions and superintegrabilty}

It is well known that the Hamiltonian of the two-dimensional Kepler problem
\begin{gather*}
 H_{K} = \frac12 \left( p_r^2 + \frac{p_{\phi}^2}{r^2} \right) + V_{K},
 \qquad V_K =-\frac{g}{r} , \qquad 0<g\in\mathbb{R},
\end{gather*}
is Hamilton--Jacobi (H-J) separable in polar coordinates $(r,\phi)$ and it is, therefore,
Liouville integrable with the angular momentum $J_2=p_\phi$ as the associate constant of motion.
Moreover, it is also known that it is a super-integrable system with the two components of the
Laplace--Runge--Lenz vector as additional integrals of motion.
Now we will prove that this property of superintegrability can be related to the existence of certain
complex functions with interesting Poisson bracket properties.

Let us denote by $M_{r j}$ and $N_{\phi j}$, $j=1,2$, the following real functions
\begin{gather*}
 M_{r1} = p_r p_\phi ,\qquad M_{r2} = g - \frac{p_\phi^2}{r} ,
\qquad \text{and}\qquad
 N_{\phi1} = \cos \phi ,\qquad N_{\phi2} = \sin \phi .
\end{gather*}
Then we have the following properties
\begin{alignat*}{4}
& \textrm{(i)}\quad&& \frac{d}{d t} M_{r1} = \{M_{r1}, H_K\} = - \lambda M_{r2}
,\qquad &&
 \frac{d}{d t} M_{r2} = \{M_{r2}, H_K\} = \lambda M_{r1} ,&
\\
& \textrm{(ii)}\quad && \frac{d}{d t} N_{\phi1} = \{N_{\phi1}, H_K\} = - \lambda N_{\phi2}
,\qquad&&
 \frac{d}{d t} N_{\phi2} = \{N_{\phi2}, H_K\} = \lambda N_{\phi1} ,&
\end{alignat*}
where $\lambda$ denote the following function
\begin{gather*}
 \lambda = \frac{p_\phi}{r^2} .
\end{gather*}
The property (ii), representing the behaviour of the angular functions $N_{\phi j}$, is true for all the central potentials~$V(r)$; but the property~(i), behaviour of the functions~$M_{r j}$, is specif\/ic of the potential of the Kepler problem.

 Consider next the complexif\/ication of the linear space of functions on the manifold and extend by bilinearity the Poisson bracket. If we denote~$M_r$ and $N_\phi$ the complex functions
\begin{gather*}
 M_r = M_{r1} + i M_{r2} ,\qquad N_\phi = N_{\phi1} + i N_{\phi2} ,
\end{gather*}
then they have the following properties
\begin{gather*}
 \{M_{r}, H_K\} = i \lambda M_{r} ,\qquad
 \{N_{\phi}, H_K\} = i \lambda N_{\phi} ,
\end{gather*}
and consequently the Poisson bracket of $M_r N_\phi^{*}$ with the Kepler Hamiltonian vanishes
\begin{gather*}
 \{M_r N_\phi^{*} , H_K\} = \{M_r , H_K\} N_\phi^{*}
 + M_r \{N_\phi^{*} , H_K\}
 = (i \lambda M_{r} ) N_\phi^{*}
 + M_r ( -i \lambda N_{\phi}^* ) = 0 .
\end{gather*}

We can summarize this result in the following proposition.

\begin{proposition} \label{proposition1}
Let us consider the Hamiltonian of the Kepler problem
\begin{gather*}
 H_{K} = \frac12 \left( p_r^2 + \frac{p_{\phi}^2}{r^2} \right) + V_{K},
 \qquad V_K =-\frac{g}{r} ,
\end{gather*}
Then, the complex function $J_{34}$ defined as
\begin{gather*}
 J_{34} = M_r N_\phi^{*}
\end{gather*}
is a $($complex$)$ constant of the motion.
\end{proposition}

Of course $J_{34}$ determines two real f\/irst-integrals
\begin{gather*}
 J_{34} = J_3 + i J_4 ,\qquad
 \{J_3, H_K \} = 0 ,\qquad \{J_4, H_K \} = 0 ,
\end{gather*}
whose coordinate expressions turn out to be
\begin{gather*}
 J_3 = {\rm Re}(J_{34}) = p_r p_\phi \cos\phi - \frac{p_\phi^2}{r} \sin\phi + g \sin\phi ,\\
 J_4 = {\rm Im}(J_{34}) = p_r p_\phi \sin\phi + \frac{p_\phi^2}{r} \cos\phi - g \cos\phi .
\end{gather*}
That is, the two functions $J_3$ and $J_4$ are just the two components of the two-dimensional
Laplace--Runge--Lenz vector.

Summarizing, we have got two interesting properties.
First, the superintegrability of the Kepler problem is directly related with the existence of two complex functions whose Poisson brackets with the Hamiltonian are proportional with a common complex factor to themselves,
and second, the two components of the Laplace--Runge--Lenz vector appear as the real and imaginary
parts of the complex f\/irst-integral of motion.
Remark that $ N_\phi$ is a~complex function of constant modulus one, while the modulus of~$ M_r$ is a polynomial of degree four in the momenta given by
\begin{gather*}
 M_r M_r^* = ( p_r p_\phi )^2 + \left( g - \frac{p_\phi^2}{r} \right)^2
 = 2 p_\phi^2 H_K + g^2 .
\end{gather*}

\subsection{Complex functions and quasi-bi-Hamiltonian structures}

Let us denote by $Y_{34}$ the (complex) Hamiltonian vector f\/ield of $J_{34}$
\begin{gather*}
 i(Y_{34}) \omega_0 = d J_{34} ,
\end{gather*}
that obviously satisf\/ies $Y_{34}(H_K) = \{H_K,J_{34}\} = 0$,
and by $Y_{r}$ and $Y_\phi$ the Hamiltonian vector f\/ields of~$M_r$ and~$N_\phi$:
\begin{gather*}
 i(Y_{r}) \omega_0 = d M_{r} ,\qquad i(Y_{\phi}) \omega_0 = d N_{\phi} .
\end{gather*}
Their local coordinate expressions are, respectively, given by
\begin{gather*}
 Y_r = p_\phi \fracpd{}{r} + \left(p_r - 2 i \frac{p_\phi}{r}\right)\fracpd{}{\phi}
 - i \left(\frac{p_\phi^2}{r^2}\right) \fracpd{}{p_r} ,
\end{gather*}
and
\begin{gather*}
 Y_\phi = (\sin\phi - i \cos\phi) \fracpd{}{p_\phi} .
\end{gather*}
Then, the vector f\/ield $Y_{34}$ appears as a linear combination of $Y_{r}$ and $Y_\phi^*$; more specif\/ically we have
\begin{gather*}
 Y_{34} = N_\phi^* Y_r + M_r Y_\phi^* = Y + Y' ,\qquad
 Y = N_\phi^* Y_r ,\qquad Y' = M_r Y_\phi^* .
\end{gather*}
The vector f\/ield $Y_{34}$ is certainly a symmetry of the Hamiltonian system $(T^*Q,\omega_0,H_K)$,
but the two vector f\/ields, $Y$ and $Y'$, are neither symmetries of the symplectic form~$\omega_0$
(that is, ${\mathcal X}_{Y} \om_0\ne 0$ and ${\mathcal X}_{Y'}\om_0\ne 0$) nor symmetries of the Hamiltonian (that is, ${\mathcal X}_{Y}H_K\ne 0$ and ${\mathcal X}_{Y'}H_K\ne 0$).
Moreover, remark that they are not symmetries of the dynamics, because
\begin{gather*}
 [Y,\Ga_K] \ne 0 ,\qquad [Y',\Ga_K] \ne 0 ,\qquad i(\Ga_{K}) \omega_0 = d H_K.
\end{gather*}
More specif\/ically:

\begin{proposition} \label{proposition2}
The Lie bracket of the dynamical vector field $\Gamma_K$ with $Y$ is given by
\begin{gather*}
 [\Ga_K, Y] = i J_{34} X_{\lambda},
\end{gather*}
where $X_{\lambda}$ is the Hamiltonian vector field of the function $\lambda$.
\end{proposition}
\begin{proof}
A direct computation leads to
\begin{gather*}
 [\Ga_K,Y] = \Ga_K( N_\phi^*) Y_r + N_\phi^* [\Ga_K,Y_r]
 = -i {\lambda} N_\phi^* Y_r + N_\phi^*\bigl({-}X_{\{H_K,M_r\}}\bigr) .
\end{gather*}
where we have used that the Lie bracket of two Hamiltonian vector f\/ields satisf\/ies $[X_f,X_g] = -X_{\{f,g\}}$.
Note also that the Hamiltonian vector f\/ield of a~product~$fg$ is given by $X_{fg}=f X_g + g X_f$,
and then the above Lie bracket becomes
\begin{gather*}
 [\Ga_K,Y] = -i {\lambda} N_\phi^* Y_r + i N_\phi^* (X_{{\lambda}M_r} )
 = -i {\lambda} N_\phi^* Y_r + i N_\phi^* ({\lambda} Y_r + M_r X_{\lambda} )
 = i ( M_r N_\phi^* )X_{\lambda}.\tag*{\qed}
\end{gather*}
\renewcommand{\qed}{}
\end{proof}

The vector f\/ield $X_{\lambda}$ on the right hand side represents an obstruction for~$Y$ to be a dynamical symmetry. Only when~${\lambda}$ be a numerical constant the vector f\/ield~$Y$ (and also~$Y'$) is a dynamical symmetry of~$\Gamma_K$.

In the following $\Omega$ will denote the complex 2-form def\/ined as
\begin{gather*}
 \Omega = dM_r \wedge dN_\phi^* .
\end{gather*}
The two complex 2-forms $\om_Y$ and $\om_Y'$ obtained by Lie derivative of $\om_0$, i.e.,
\begin{gather*}
 {\mathcal L}_{Y} \om_0 = \om_Y ,\qquad {\mathcal L}_{Y'} \om_0 = \om_Y' ,
\end{gather*}
are such
\begin{gather*}
 {\mathcal L}_{Y} \om_0 = i_Y (d\om_0) + d(i_Y\om_0) = d(i_Y\om_0) = d (N_\phi^* dM_r )
 =-\Omega,
\\
 {\mathcal L}_{Y'} \om_0 = i_{Y'} (d\om_0) + d(i_{Y'}\om_0) = d(i_{Y'}\om_0) = d (M_r dN_\phi^* ) = \Omega.
\end{gather*}

Using the preceding results we can prove:

\begin{proposition} \label{proposition3}
The Hamiltonian vector field $\Ga_K$ of the Kepler problem is a quasi-Hamiltonian system with
respect to the complex $2$-form $\Omega$.
\end{proposition}

\begin{proof}
 The contraction of the vector f\/ield $ \Ga_K$ with the complex 2-form $\Omega$ gives
\begin{gather*}
 i(\Ga_K) \Omega = \Ga_K (M_r) dN_\phi^* - \Ga_K(N_\phi^*) dM_r ,
\end{gather*}
and recalling that
\begin{gather*}
 \Ga_K (M_r) = \{M_{r}, H_K\} = i \lambda M_{r}
,\qquad
 \Ga_K (N_\phi^*) = \{N_\phi^*, H_K\} = - i \lambda N_\phi^* ,
\end{gather*}
we arrive to
\begin{gather*}
 i(\Ga_K) \Omega = ( i \lambda M_{r}) dN_\phi^* + (i \lambda N_\phi^*) dM_r
 = i \lambda d(M_rN_\phi^*) .\tag*{\qed}
\end{gather*}
\renewcommand{\qed}{}
\end{proof}

The complex 2-form $\Omega$ can be written as
\begin{gather*}
 \Omega = \Omega_1 + i \Omega_2
\end{gather*}
 where the two real 2-forms, $\Omega_1 = {\rm Re}(\Omega)$ and $\Omega_2={\rm Im}(\Omega)$,
 take the form
\begin{gather*}
 \Omega_1 = d M_{r1} \wedge dN_{\phi1} + d M_{r2} \wedge dN_{\phi2}
= d (p_r p_\phi ) \wedge d ( \cos\phi )
 + d\left(g- \frac{p_\phi^2}{r}\right) \wedge d ( \sin\phi ) \\
\hphantom{\Omega_1}{}
 = \al_{12} dr \wedge d\phi + \al_{23} d\phi \wedge dp_r + \al_{24} d\phi \wedge dp_\phi,
\\
 \Omega_2 = -d M_{r1} \wedge dN_{\phi2} + d M_{r2} \wedge dN_{\phi1}
= - d (p_r p_\phi ) \wedge d ( \sin\phi )
 + d\left(g- \frac{p_\phi^2}{r}\right) \wedge d ( \cos\phi ) \\
\hphantom{\Omega_2}{}
= \beta_{12} dr \wedge d\phi + \beta_{23} d\phi \wedge dp_r + \beta_{24} d\phi \wedge dp_\phi
\end{gather*}
with $\alpha_{ij}$ and $\beta_{ij}$ being given by
\begin{gather*}
 \al_{12} = \left( \frac{p_\phi^2}{r^2}\right)\cos\phi ,\qquad
 \al_{23} = p_\phi \sin\phi ,\qquad
 \al_{24} = p_r\sin\phi + 2 \left( \frac{p_\phi}{r}\right) \cos\phi ,
\end{gather*}
and
\begin{gather*}
 \beta_{12} = -\left( \frac{p_\phi^2}{r^2}\right) \sin\phi ,\qquad
 \beta_{23} = p_\phi \cos\phi ,\qquad
 \beta_{24} = p_r\cos\phi - 2 \left( \frac{p_\phi}{r}\right) \sin\phi .
\end{gather*}
Then we have
\begin{gather*}
 i(\Ga_K) \Omega_1 = -\lambda dJ_4 ,\qquad
 i(\Ga_K) \Omega_2 = \lambda dJ_3 ,
\end{gather*}
what means that $\Ga_K$ is also quasi-bi-Hamiltonian with respect to the two real 2-forms $(\om_0,\Omega_1)$ or $(\om_0,\Omega_2)$.

Remark that the complex 2-form $\Omega$ is well def\/ined but it is not symplectic.
In fact, from the above expressions in coordinates we have
$\Omega_1 \wedge \Omega_1 = 0$,
$\Omega_2 \wedge \Omega_2 = 0$, and
$\Omega_1 \wedge \Omega_2 = 0$,
and therefore we obtain
\begin{gather*}
 \Omega \wedge \Omega = (\Omega_1 \wedge \Omega_1 - \Omega_2 \wedge \Omega_2 ) + 2 i \Omega_1 \wedge \Omega_2 = 0 .
\end{gather*}

The distribution def\/ined by the kernel of $\Omega$, that is two-dimensional, is given by
\begin{gather*}
 \operatorname{Ker} \Omega = \big\{ f_1 Z_1+ f_2 Z_2 \, |\, f_1,f_2\colon \mathbb{R}^2\times\mathbb{R}^2\to\mathbb{C} \big\},
\end{gather*}
where the vector f\/ields $Z_1$ and $Z_2$ are
\begin{gather*}
Z_1 = (\al_{23} + i \be_{23}) \fracpd{}{r} + (\al_{12} + i \be_{12}) \fracpd{}{p_r} ,\qquad
Z_2 = (\al_{24} + i \be_{24}) \fracpd{}{r} + (\al_{12} + i \be_{12}) \fracpd{}{p_\phi} .
\end{gather*}
Therefore it satisf\/ies
\begin{gather*}
 [\operatorname{Ker} \Omega,\Ga_K] \subset \operatorname{Ker} \Omega.
\end{gather*}
 That is, $\Ga_K$ preserves the distribution $\operatorname{Ker} \Omega$.

If $Y_3$ and $Y_4$ are the Hamiltonian vector f\/ields (with respect to the canonical symplectic form~$\om_0$) of the f\/irst integrals $J_3$ and $J_4$, then the dynamical vector f\/ield~$\Ga_K$ is orthogonal to~$Y_4$ with respect to the structure~$\Omega_1$ and it is also orthogonal to $Y_3$ with respect to the structure~$\Omega_2$, that is,
\begin{gather*}
 i(\Ga_K) i(Y_4) \Omega_1 = 0 , \qquad i(\Ga_K) i(Y_3) \Omega_2 = 0 .
\end{gather*}

Just to close the section we remark that had we applied this technique to the isotropic
two-dimensional harmonic oscillator with frequency $\alpha$ we had obtained the function~$M_r$ as
\begin{gather*}
 M_r = \left(\frac{2}{r} p_r p_\phi \right) + i \left( p_r^2 - \frac{p_\phi^2}{r^2} + \al^2 r^2 \right),
\end{gather*}
(the angular function $N_\phi$ would be the same) and the constants so obtained are but the components of the Fradkin tensor.
This shows that the harmonic oscillator is an example of dynamical system both bi-Hamiltonian and quasi-bi-Hamiltonian.

\subsection{Recursion operators and some comments} \label{section23}

The bi-Hamiltonian structure $(\om_0,\Omega)$ determines a complex recursion operator $
R$ def\/ined as
\begin{gather*}
 \Omega(X,Y) = \om_0(RX,Y) ,\qquad \forall\, X,Y \in \mathfrak{X}(T^*Q) .
\end{gather*}
But as $\Omega$ and $R$ are complex, we can introduce two real recursion operator
$R_{1}$ and $R_{2}$ def\/ined as
\begin{gather*}
 \Omega_{1}(X,Y) = \om_0(R_{1}X,Y) ,\qquad \Omega_{2}(X,Y) = \om_0(R_{2}X,Y) .
\end{gather*}
We recall that $\wh{\om_0}$ is the map $\wh{\om_0}\colon \mathfrak{X}(T^*Q)\to\wedge^1(T^*Q)$
given by contraction, that is $\wh{\om_0}(X) = i(X)\om_0$, and then the nondegenerate character
of $\om_0$ means that the map $\wh{\om_0}$ is a bijection.
Using this notation we can write the two operators $R_{1}$ and $R_{2}$ as follows
\begin{gather*}
 R_{1} = \wh{\om_0}^{-1} \circ \wh{\Omega_{1}} ,\qquad
 R_{2} = \wh{\om_0}^{-1} \circ \wh{\Omega_{2}} .
\end{gather*}
Then we have the following properties
\begin{itemize}\itemsep=0pt
\item[(i)] The coordinates expressions of $R_{1}$ and $R_{2}$ are
\begin{gather*}
 R_{1} = -\al_{12} \fracpd{}{p_\phi} \otimes dr
 + \left[\al_{23} \fracpd{}{r} + \al_{24}\fracpd{}{\phi}+ \al_{12}\fracpd{}{p_r}\right] \otimes d\phi\\
 \hphantom{R_{1} =}{}
 + \al_{23} \fracpd{}{p_\phi} \otimes dp_r
 + \al_{24} \fracpd{}{p_\phi} \otimes dp_\phi
\end{gather*}
and
\begin{gather*}
 R_{2} = -\be_{12} \fracpd{}{p_\phi} \otimes dr
 + \left[\be_{23} \fracpd{}{r} + \be_{24}\fracpd{}{\phi}+ \be_{12}\fracpd{}{p_r}\right] \otimes d\phi\\
 \hphantom{R_{2} =}{}
 + \be_{23} \fracpd{}{p_\phi} \otimes dp_r
 + \be_{24} \fracpd{}{p_\phi} \otimes dp_\phi .
\end{gather*}

\item[(ii)] $R_{1}$ and $R_{2}$ have two dif\/ferent eigenvalues doubly degenerate and one of them is null (that is, $\lambda_1=\lambda_2=0$, $\lambda_3=\lambda_4\ne 0$).
Therefore we have
\begin{gather*}
 \det [R_{1}] =\det [R_{2}] = 0,
\end{gather*}
what is a consequence of the singular character of $\Omega_1$ and $\Omega_2$.
\end{itemize}

 We close this section summarizing the situation we have arrived.
We have f\/irst introduced two complex functions, $M_r$ and $N_\phi$, mainly because
of the behaviour of their Poisson brackets.
Then we have proved that they are interesting for two reasons: f\/irst because they determine the existence of superintegrability (existence of additional constants of motion) and second because they determine quasi-bi-Hamiltonian structures (f\/irst complex $(\om_0,\Omega)$ and then real $(\om_0,\Omega_1,\Omega_2)$).

 Concerning the f\/irst point, in this case the additional constants of motion are just the components of the Runge--Lenz vector (that have been highly studied making use of dif\/ferent approaches).
Now we have arrived to a new property: they can also be obtained as a consequence of this complex formalism.

 Concerning the second point, the two complex functions $M_r$ and $N_\phi$ determine the
above mentioned geometric structures (f\/irst complex and then real) but unfortunately they are de\-ge\-nerated
(we recall that $\Omega_1 \wedge \Omega_1 = 0$, $\Omega_2 \wedge \Omega_2 = 0$).
This can be considered as a limitation of these geometric structures.
If a bi-Hamiltonians structure satisf\/ies all the appropriate properties (that is, symplectic forms, vanishing of the Nijenhuis torsion of the recursion opera\-tor~$R$, diagonalizable recursion opera\-tor~$R$ with functionally independent real eigenvalues) then it determines the Liouville integrability of the system.
In fact, the aim of the approach presented in this paper is not to prove the integrability of a system as consequence of a bi-Hamiltonian structure as we start with a system known to be not only integrable but also superintegrable.
The existence of $(\Omega_1,\Omega_2)$ must be considered, not as a method for arriving to the integrability of the system, but as a~new and interesting property of the Kepler problem
(for the moment only of the two-dimensional system, the generalization to the three-dimensional case must be considered as an open question).

\section[New complex functions and new quasi-bi-Hamiltonian structures]{New complex functions and new quasi-bi-Hamiltonian\\ structures}\label{section3}

In this section we will study the existence of new bi-Hamiltonian structures for the Kepler dynamics by making use of parabolic coordinates~$(a,b)$ def\/ined as
\begin{gather*}
 x = \tfrac{1}{2}\big(a^2 - b^2\big) ,\qquad y = a b .
\end{gather*}
Of course all previous results can be translated to this new language in such a way that the functions $M_{r1} $ and $M_{r2}$ are now given by
\begin{gather*}
 M_{r1} =\frac{(a p_b - b p_a) (a p_a + b p_b)}{\sqrt{a^2+b^2}} ,\qquad
 M_{r2} = \frac{(a p_b - b p_a)^2}{\sqrt{a^2+b^2}} - g ,
\end{gather*}
while functions $N_{\phi1}$ and $N_{\phi2}$ become
\begin{gather*}
 N_{\phi 1} = \frac{a^2 - b^2}{\sqrt{a^2+b^2}} ,\qquad N_{\phi 2} = \frac{2 a b}{\sqrt{a^2+b^2}} .
\end{gather*}
But the important point is that the behaviour of the Kepler Hamiltonian in these coordinates will permit us to obtain new results dif\/ferent from the previous ones.

\subsection{Complex functions and superintegrability}

The general form of a natural Euclidean Hamiltonian is
\begin{gather*}
 H = \frac{1}{2m} \left(\frac{p_a^2 + p_b^2}{a^2 + b^2}\right) + V(a,b) .
\end{gather*}
in such a way that if the potential $V$ is of the form
\begin{gather*}
 V(a,b) = \frac{A(a) + B(b)}{a^2 + b^2} ,
\end{gather*}
then the Hamiltonian is Hamilton--Jacobi separable and it is, therefore,
 Liouville integrable with the following quadratic function
\begin{gather*}
 J_2 = \frac{1}{(a^2 + b^2)}(a p_b - b p_a)(a p_b + b p_a) +
 2 \left(\frac{a^2B-b^2A}{a^2+b^2}\right)
\end{gather*}
as the second constant of motion (the f\/irst one is the Hamiltonian itself).

For simplifying the following expressions we introduce the following notation:
\begin{gather*}
 J = a p_b - b p_a ,\qquad
 P_x = \frac{a p_a - b p_b}{a^2+b^2} , \qquad
 P_y = \frac{a p_b + b p_a}{a^2+b^2} . 
\end{gather*}

The Hamiltonian of the Kepler problem when written in parabolic coordinates is
\begin{gather}
 H_{K} = \frac12 \left(\frac{ p_a^2 + p_b^2}{a^2+b^2} \right) + V_{K},
 \qquad V_K =-\frac{g}{a^2+b^2} ,\label{HKparab}
\end{gather}
and the Kepler dynamics is given by the following vector f\/ield
\begin{gather*}
 \Ga_K = \left(\frac{p_a}{a^2+b^2}\right)\fracpd{}{a} +\left(\frac{p_b}{a^2+b^2}\right)\fracpd{}{b}
 + \left(\frac{p_a^2+p_b^2 - 2 g}{(a^2+b^2)^2}\right) a \fracpd{}{p_a}
 + \left(\frac{p_a^2+p_b^2 - 2 g}{(a^2+b^2)^2}\right) b \fracpd{}{p_b} ,
\end{gather*}
in such a way that, as $\omega_0=da \wedge dp_a + db \wedge dp_b$, we have
\begin{gather*}
 i(\Ga_K) (da \wedge dp_a + db \wedge dp_b) = dH_K .
\end{gather*}
The Hamiltonian $H_K$ is Hamilton--Jacobi separable in coordinates $(a,b)$ and the associated quadratic constant of motion is the component $R_x$ of the Laplace--Runge--Lenz vector
\begin{gather*}
 R_x = J P_y - g \left(\frac{a^2 - b^2}{a^2+b^2} \right) .
\end{gather*}

Let us now introduce the functions $M_{a j}$ and $M_{b j}$, $j=1,2$, def\/ined by
\begin{gather*}
 M_{a1} = \frac{J p_a}{\sqrt{a^2+b^2}} ,\qquad M_{a2} = \frac{2 g a- J p_b}{\sqrt{a^2+b^2}} ,
\end{gather*}
and
\begin{gather*}
 M_{b1} = \frac{J p_b}{\sqrt{a^2+b^2}} ,\qquad M_{b2} = \frac{2 g b + J p_a}{\sqrt{a^2+b^2}} .
\end{gather*}
Then, the important property is that the Poisson bracket of the function $M_{a1}$ with $H_K$
is proportional to $M_{a2}$ while the Poisson bracket of $M_{a2}$ with~$H_K$ is proportional
to $M_{a1}$, but with the opposite sign:
\begin{gather*}
 \{M_{a1}, H_K\} = - \lambda M_{a2} ,\qquad
 \{M_{a2}, H_K\} = \lambda M_{a1} ,
\end{gather*}
and the same is true for the functions $M_{b1}$ and $M_{b2}$,
\begin{gather*}
 \{M_{b1}, H_K\} = - \lambda M_{b2} ,\qquad
 \{M_{b2}, H_K\} = \lambda M_{b1} ,
\end{gather*}
where now $\lambda$ denotes the following function
\begin{gather*}
 \lambda = \frac{a p_b - b p_a}{(a^2+b^2)^2} .
\end{gather*}
Therefore the two complex functions $M_a$ and $M_b$ def\/ined as
\begin{gather*}
 M_a = M_{a1} + i M_{a2} ,\qquad M_b = M_{b1} + i M_{b2} ,
\end{gather*}
satisfy
\begin{gather*}
 \{M_a, H_K\} = i \lambda M_{a} ,\qquad
 \{M_b, H_K\} = i \lambda M_{b} .
\end{gather*}

\begin{proposition} \label{proposition4}
The complex function $K_{34}$ defined as
\begin{gather*}
 K_{34} = M_a M_b^{*}
\end{gather*}
is a $($complex$)$ constant of the motion for the dynamics of the Kepler problem described by the Hamiltonian~\eqref{HKparab}.
\end{proposition}

We omit the proof because it is quite similar to the proof of the previous Proposition~\ref{proposition1}.

Note that the modulus of the complex functions $M_a$ and $M_b$ are given by
\begin{gather*}
 M_a M_a^* = 2 \big( J^2 H_K - g R_x + g^2\big) ,\qquad
 M_b M_b^* = 2 \big( J^2 H_K + g R_x + g^2\big),
\end{gather*}
and then
\begin{gather*}
 M_a M_a^* + M_b M_b^* = 4 \big( J^2 H_K + g^2\big).
\end{gather*}

Of course the complex function $K_{34}$ determines two real functions that are f\/irst integrals
for the Kepler problem
\begin{gather*}
 K_{34} = K_3 + i K_4 ,\qquad
\{K_3, H_K\} = 0 ,\qquad
\{K_4, H_K\} = 0 ,
\end{gather*}
with $K_3$ and $K_4$ given by
\begin{gather*}
 K_3 = {\rm Re}(K_{34}) = M_{a1}M_{b1} + M_{a2}M_{b2}
 = J P_x + g \left(\frac{2 a b}{a^2+b^2}\right) , \cr
 K_4 = {\rm Im}(K_{34}) = M_{a2}M_{b1} - M_{a1}M_{b2}
= - 2 J^2 H_K .
\end{gather*}
Remark that the function $K_3$ is the other Laplace--Runge--Lenz constant, while $K_4$,
that is a~fourth order polynomial in the momenta, determines the angular momentum $J$ as a factor.

\subsection{Complex functions and quasi-bi-Hamiltonian structures}

First we recall that the complex functions $M_a$ y $M_b$ are given by
\begin{gather*}
 M_{a} = \left( \frac{J p_a}{\sqrt{a^2+b^2}} \right) + i \left(\frac{2 g a- J p_b}{\sqrt{a^2+b^2}} \right)
 , \qquad
 M_{b} = \left( \frac{J p_b}{\sqrt{a^2+b^2}} \right) + i \left(\frac{2 g b + J p_a}{\sqrt{a^2+b^2}}\right) .
\end{gather*}

Let us now denote by $Z_{34}$ the Hamiltonian vector f\/ield of the function $K_{34}$, i.e.,
$ i(Z_{34}) \omega_0 = d Z_{34} $, such that $Z_{34}(H_K) = 0$,
and by $Z_{a}$ and $Z_{b}$ the Hamiltonian vector f\/ields of the complex functions $M_a$ and $M_b$, that is,
\begin{gather*}
 i(Z_{a}) \omega_0 = d M_{a} ,\qquad i(Z_{b}) \omega_0 = d M_{b} .
\end{gather*}
Their coordinate expressions are given by
\begin{gather*}
 Z_a = \left(\fracpd{M_a}{p_a}\right)\fracpd{}{a} + \left(\fracpd{M_a}{p_b}\right)\fracpd{}{b} - \left(\fracpd{M_a}{a}\right)\fracpd{}{p_a} - \left(\fracpd{M_a}{b}\right)\fracpd{}{p_b}
 = Z_{a1} + i Z_{a2} ,
\end{gather*}
with $Z_{a1}$ and $Z_{a2}$ given by
\begin{gather*}
 Z_{a1} = \frac{1}{\sqrt{a^2+b^2}} \left( (a p_b - 2 b p_a ) \fracpd{}{a}
 + a p_a \fracpd{}{b} - \frac{(a p_a + b p_b)}{a^2+b^2}
 \left( b p_a \fracpd{}{p_a} - a p_a \fracpd{}{p_b} \right) \right),
\\
 Z_{a2} = \frac{1}{\sqrt{a^2+b^2}} \left( b p_b \fracpd{}{a}
 + (b p_a -2 a p_b ) \fracpd{}{b} - \frac{(a p_a + b p_b)p_b-2 g b}{a^2+b^2}
 \left( -b \fracpd{}{p_a} + a \fracpd{}{p_b} \right)\right),
\end{gather*}
and
\begin{gather*}
 Z_b = \left(\fracpd{M_b}{p_a}\right)\fracpd{}{a}+ \left(\fracpd{M_b}{p_b}\right)\fracpd{}{b}- \left(\fracpd{M_b}{a}\right)\fracpd{}{p_a}- \left(\fracpd{M_b}{b}\right)\fracpd{}{p_b}
 = Z_{b1} + i Z_{b2} ,
\end{gather*}
with $Z_{b1}$ and $Z_{b2}$ given by
\begin{gather*}
 Z_{b1} = \frac{1}{\sqrt{a^2+b^2}} \left( - b p_b \fracpd{}{a}
 + (2 a p_b - b p_a ) \fracpd{}{b} - \frac{(a p_a + b p_b)}{a^2+b^2}
 \left( b p_b \fracpd{}{p_a} -a p_b \fracpd{}{p_b} \right) \right),
\\
 Z_{b2} = \frac{1}{\sqrt{a^2+b^2}} \left( (a p_b - 2 b p_a ) \fracpd{}{a}
 + a p_a \fracpd{}{b} - \frac{(a p_a + b p_b)p_a-2 g a}{a^2+b^2}
 \left( b \fracpd{}{p_a} -a \fracpd{}{p_b} \right) \right).
\end{gather*}
Now recalling that
\begin{gather*}
 dZ_{34} = d (M_a M_b^* ) = M_b^* d (M_a ) + M_a d (M_b^* ),
\end{gather*}
we obtain
\begin{gather*}
 Z_{34} = M_b^* Z_a + M_a Z_b^* = Z + Z',\qquad \textrm{where}\qquad
 Z = M_b^* Z_a ,\qquad Z' = M_a Z_b^* .
\end{gather*}
The following proposition is similar to that of Proposition~\ref{proposition2} and we omit the proof.

\begin{proposition} \label{proposition5}
The Lie bracket of the Kepler dynamical vector field $\Gamma_K$ with the vector field~$Z$ is given by
\begin{gather*}
 [\Ga_K, Z] = i K_{34} X_{\lambda},
\end{gather*}
where $X_{\lambda}$ is the Hamiltonian vector field of~$\lambda$ solution of the equation $ i(X_{\lambda}) \om_0 = d \lambda$.
\end{proposition}

In the following we will denote by $\Omega_{ab}$ the complex 2-form def\/ined as
\begin{gather*}
 \Omega_{ab} = dM_a \wedge dM_b^* \\
 \hphantom{\Omega_{ab}}{}
 = d\left[ \left( \frac{J p_a}{\sqrt{a^2+b^2}} \right) + i \left(\frac{2 g a- J p_b}{\sqrt{a^2+b^2}} \right) \right] \wedge
 d\left[\left( \frac{J p_b}{\sqrt{a^2+b^2}} \right) - i \left(\frac{2 g b + J p_a}{\sqrt{a^2+b^2}}\right) \right] .
\end{gather*}
Then the two 2-forms $\om_Z$ and $\om_Z'$ obtained by Lie derivation of $\om_0$ with respect to $Z$ and $Z'$ are given by
\begin{gather*}
 {\mathcal L}_{Z} \om_0 = \om_Z = -\Omega_{ab} ,\qquad
 {\mathcal L}_{Z'} \om_0 = \om_Z' = \Omega_{ab}.
\end{gather*}

\begin{proposition} \label{proposition6}
The Hamiltonian vector field $\Ga_K$ of the Kepler problem is quasi-Hamiltonian with respect to the complex
$2$-form $\Omega_{ab}$.
\end{proposition}

\begin{proof}
 This can be proved by a direct computation
\begin{gather*}
 i(\Ga_K) \Omega_{ab} = \Ga_K (M_a) dM_b^* - \Ga_K(M_b^*) dM_a
 = ( i \lambda M_{a}) dM_b^* + (i \lambda M_b^*) dM_a
 = i \lambda d(M_a M_b^*).\!\!\!\!\!\!\tag*{\qed}
\end{gather*}
\renewcommand{\qed}{}
\end{proof}

 The complex 2-form $ \Omega_{ab} $ can be decomposed as
\begin{gather*}
 \Omega_{ab} = \Omega_{ab1} + i \Omega_{ab2},
\end{gather*}
 where the two real 2-forms, $\Omega_{ab1} = {\rm Re}(\Omega_{ab} )$ and
 $\Omega_{ab2}={\rm Im}(\Omega_{ab})$, take the form
\begin{gather*}
 \Omega_{ab1} = d M_{a1} \wedge dM_{b1} + d M_{a2} \wedge dM_{b2} ,\qquad
 \Omega_{ab2} = -d M_{a1} \wedge dM_{b2} + d M_{a2} \wedge dM_{b1} ,
\end{gather*}
Considering the real and imaginary parts we obtain
\begin{gather*}
 i(\Ga_K) \Omega_{ab1} = -\lambda dK_4 ,\qquad
 i(\Ga_K) \Omega_{ab2} = \lambda dK_3 ,
\end{gather*}
what means that $\Ga_K$ is also quasi-bi-Hamiltonian with respect to the two real 2-forms $(\om_0, \Omega_{ab1})$ and $(\om_0, \Omega_{ab2})$.

Once more we obtain that the factor $\lambda$ determines that the system is quasi-bi-Hamiltonian instead of just bi-Hamiltonian.

 The complex 2-form $\Omega_{ab}$ is well def\/ined but it is not symplectic. The kernel is two-dimensional and it is invariant under the action of $\Ga_K$
\begin{gather*}
 [\operatorname{Ker} \Omega_{ab},\Ga_K] \subset \operatorname{Ker} \Omega_{ab}.
\end{gather*}

The coordinate expressions of $\Omega_{ab1}$ and $ \Omega_{ab2}$ are
\begin{gather*}
 \Omega_{ab1} = \frac{2 J}{(a^2\!+b^2)^2}
 (\al_{13} da \wedge dp_a + \al_{14} da \wedge dp_b + \al_{23} db \wedge dp_a + \al_{24} db \wedge dp_b + \al_{34} dp_a \wedge dp_b ) ,
\\
 \Omega_{ab2} = \frac{2 g}{(a^2\!+b^2)^2}
 ( \be_{13} da \wedge dp_a + \be_{14} da \wedge dp_b + \be_{23} db \wedge dp_a + \be_{24} db \wedge dp_b ) ,
\end{gather*}
with $\alpha_{ij}$ and $\beta_{ij}$ being given by
\begin{alignat*}{3}
& \al_{13} = \bigl(2 g b - a p_a p_b - b p_b^2 \bigr) b , \qquad&&
 \al_{14} = -\bigl(2 g a - a p_a^2 - b p_a p_b \bigr) b , & \\
& \al_{23} = \bigl(-2 g b + a p_a p_b + b p_b^2 \bigr) a , \qquad&&
 \al_{24} = \bigl(2 g a - a p_a^2 - b p_a p_b \bigr) a , &\\
& \al_{34} = 2 J \big(a^2+b^2\big) ,&&&
\end{alignat*}
and
\begin{alignat*}{3}
& \be_{13} = \bigl( 2 a b p_a - a^2 p_b - b^2 p_b \bigr) b , \qquad&&
 \be_{14} = \bigl( 2 a b p_b - a^2 p_a - b^2 p_a \bigr) b ,& \\
 &\be_{23} = \bigl(-2 a b p_a + a^2 p_b + b^2 p_b \bigr) a , \qquad&&
 \be_{24} = \bigl(-2 a b p_b - a^2 p_a - b^2 p_a \bigr) a . &
\end{alignat*}

We close this section with the following properties:
\begin{itemize}\itemsep=0pt
\item[(i)] The two real 2-forms are not symplectic. In fact we have verif\/ied that
$\Omega_1 \wedge \Omega_1 = 0$,
$\Omega_2 \wedge \Omega_2 = 0$, and also
$\Omega_1 \wedge \Omega_2 = 0$.

 \item[(ii)] These two 2-forms, $\Omega_{ab1}$ and $\Omega_{ab2}$, determine two recursion operators
($(1,1)$ tensor f\/ields) $R_{ab1}$ and $R_{ab2}$ def\/ined as
\begin{gather*}
 \Omega_{ab1}(X,Y) = \om_0(R_{ab1}X,Y) ,\qquad \Omega_{ab2}(X,Y) = \om_0(R_{ab2}X,Y) ,
\end{gather*}
or in an equivalent way
\begin{gather*}
 R_{ab1} = \wh{\om_0}^{-1} \circ \wh{\Omega_{ab1}} ,\qquad
 R_{ab2} = \wh{\om_0}^{-1} \circ \wh{\Omega_{ab2}} .
\end{gather*}
As in Section~\ref{section23}, a consequence of the singular character of $\Omega_{ab1}$ and $\Omega_{ab2}$ is that
\begin{gather*}
 \det [R_{ab1}] =\det [R_{ab2}] = 0.
\end{gather*}

\item[(iii)]
If we denote by $Z_3$ and $Z_4$ the Hamiltonian vector f\/ields (with respect to the canonical
symplectic form~$\om_0$) of the integrals~$K_3$ and~$K_4$, then the dynamical vector f\/ield
$\Ga_K$ is orthogonal to $Z_4$ with respect to the structure $\Omega_1$ and it is also orthogonal
to $Z_3$ with respect to the structure~$\Omega_2$, that is,
\begin{gather*}
 i(\Ga_K) i(Y_4) \Omega_{ab1} = 0 , \qquad i(\Ga_K) i(Y_3) \Omega_{ab2} = 0 .
\end{gather*}
\end{itemize}

\section{Final comments}\label{section4}

The Kepler problem is separable in two dif\/ferent coordinate systems, and because of this, it is superintegrable with quadratic integrals of motion.
Now we have proved that this super\-integrability is directly related with the existence of certain complex functions possessing very interesting Poisson bracket properties and also that these functions are also related with the existence of quasi-bi-Hamiltonian structures.

We f\/inalize with some questions for future work.
First, as stated in the Introduction, quasi-bi-Hamiltonian structures is a matter that still remain as slightly studied (in contrast to the bi-Hamiltonian systems); so the particular case of the Kepler problem can be a good motivation to undertake a better study of these systems.
Second, the existence of these complex functions is not a specif\/ic characteristic of the Kepler problem;
in fact, it has been proved that the superintegrability of certain systems recently studied
(as the isotonic oscillator or the TTW or PW systems)
\cite{CarRanSa05, Ra12JPaTTW, Ra13JPaPW, Ra14JPaTTWk, Ra15PhysLet, RaRodSant10}
is also related with such a class of complex functions.
Therefore, an interesting open question is whether these other superintegrable systems are also endowed with bi-Hamiltonian structures or with quasi-bi-Hamiltonian structures; the starting point for this study must be a deeper analysis of the properties of these complex functions making use of the geometric (symplectic) formalism as an approach.

\subsection*{Acknowledgments}

This work has been supported by the research projects MTM--2012--33575 (MICINN, Madrid) and DGA-E24/1 (DGA, Zaragoza).

\pdfbookmark[1]{References}{ref}
\LastPageEnding

\end{document}